\documentclass[%
 reprint,
 superscriptaddress,
 nofootinbib,
 nobibnotes,
 amsmath,amssymb,
 aps,
 prl,
 floatfix,
]{revtex4-2}

\usepackage{amsthm}
\usepackage{graphicx}
\graphicspath{{figs/}}
\usepackage{dcolumn}
\usepackage{bm}
\usepackage{mathtools}
\usepackage{xcolor}
\usepackage[colorlinks=true,linkcolor=blue,citecolor=blue,urlcolor=blue]{hyperref}

\newcommand{\Sss}{S_{\mathrm{ss}}}
\newcommand{\Firr}{F_{\mathrm{irr}}}
\newcommand{\Tr}{\operatorname{Tr}}
\newcommand{\RR}{\mathbb{R}}
\newcommand{\E}{\mathbb{E}}
\DeclareMathOperator{\diag}{diag}

\newcommand{\Ie}{I^{\!\star}}

\newcommand{\plsection}[1]{\par\medskip\noindent\textit{#1.}---\,\ignorespaces}
\newcommand{\plsubsection}[1]{\par\smallskip\noindent\textit{#1.}---\,\ignorespaces}

\newtheorem{lemma}{Lemma}
\newtheorem{remark}{Remark}

\begin{document}

\preprint{APS/123-QED}

\title{Parity Selection Rule for Information and Dissipation in Driven Steady States}%

\author{Mengqi Li}
\affiliation{%
 School of Electronics and Information, Northwestern Polytechnical University, Xi'an, Shaanxi 710129, China%
}%

\author{Lixin Li}%
 \thanks{Contact author: lilixin@nwpu.edu.cn}%
\affiliation{%
 School of Electronics and Information, Northwestern Polytechnical University, Xi'an, Shaanxi 710129, China%
}%

\author{Wensheng Lin}%
\affiliation{%
 School of Electronics and Information, Northwestern Polytechnical University, Xi'an, Shaanxi 710129, China%
}%

\author{Zhu Han}
\affiliation{%
 Department of Electrical and Computer Engineering, University of Houston, Houston, TX 77004, USA%
}%

\date{\today}

\begin{abstract}
Tight equalities between symmetric information and entropy production in driven steady states remain elusive. We show that they are forbidden by a parity selection rule for rotation-driven linear nonequilibrium steady states. Whenever the relaxation and diffusion matrices commute, the snapshot mutual information between two time slices is exactly even under drive reversal, and parity violation rises linearly in the commutator norm when alignment is broken. Full isotropy strengthens this to drive-independence, and the planar mutual information takes the closed-form value of about 0.145 nats. Under the same alignment, the entropy production is exactly quadratic in the drive, and its prefactor admits an explicit closed form in the traces and determinant of the two matrices. The orthogonality of even and odd sectors leaves only one-sided thermodynamic-uncertainty bounds. The rule rests on the rotational symmetry of the drift alone and survives heavy-tailed isotropic stable noise with tail index below two, where variance-based bounds become vacuous. A falsifiable test is proposed on an electrical Brownian gyrator augmented for independent drive control with circuit-level stable-noise injection.
\end{abstract}

\maketitle

\plsection{Introduction}%
The conversion of one-sided information-thermodynamic inequalities into tight equalities would identify entropy production with an information cost in a single relation, sharpening Maxwell-demon accounts of feedback and dissipation~\cite{Horowitz2014,Ito2013,Sagawa2013,Seifert2012,Parrondo2015,Sagawa2008,Sagawa2010,Mandal2012,Hartich2014,Allahverdyan2009} and the Szil\'ard--Landauer bounds verified directly in experiment~\cite{Landauer1961,Bennett1982,Toyabe2010,Berut2012}. These statements sit within the stochastic-thermodynamic framework of fluctuation theorems and entropy-production identities~\cite{Jarzynski1997,Crooks1999,Sekimoto1998,Seifert2005,Esposito2010,Lebowitz1999,Schnakenberg1976}. A recent identity due to Cho and Jeong~\cite{Cho2025} achieves such a tight equality for a time-antisymmetric information measure, but for the natural symmetric information content of a steady state no tight equality has been established. The obstruction has been read as a missing construction. This Letter argues that it follows from a single symmetry of the driven dynamics.

For a rotation-driven linear nonequilibrium steady state, reversal of the drive frequency $\omega \to -\omega$ sorts every coordinate-free steady-state functional into a parity-even and a parity-odd sector, and the two are orthogonal. The symmetric snapshot mutual information between two time slices lives in the even sector and is, under structural alignment of the relaxation and diffusion matrices, exactly even in $\omega$; full isotropy strengthens this to exact $\omega$-independence at the closed-form value $-(n/2)\log(1-e^{-2at})$. The entropy production is \emph{odd-rooted}, with directional content time-antisymmetric in the standard path-action sense, and at the isotropic point it reduces to the closed-form quadratic $\Sigma = 2\omega^2/a$. No tight equality can equate an even information functional with an odd-rooted entropy production at every $\omega$, and what symmetry permits is only the one-sided bound carried by thermodynamic uncertainty relations~\cite{Barato2015,Gingrich2016,Horowitz2017,Hasegawa2019,Koyuk2020,Falasco2020,Manikandan2020,Skinner2021,Seifert2012}. We call this the \textit{parity selection rule}.

The variance-based thermodynamic uncertainty relation breaks down outside the Gaussian regime. For isotropic $\alpha$-stable driving with tail index $\gamma \in (0, 2)$, the circulation observable that probes the dissipative current inherits the heavy tail and has infinite variance, so the variance bound has a divergent left-hand side and is vacuous~\cite{Vo2025}. Under full isotropy, the parity rule remains exact across the same regime, because it is a symmetry of the drift rather than a moment of the noise; a sub-Gaussian construction controls finiteness through the logarithmic moments of the stable subordinator. The construction settles finiteness and $\omega$-independence of the snapshot mutual information without requiring a finite stationary covariance, while the absolute scale of the entropy production in the heavy-tailed regime requires a tempering regularization of the jump kernel. The rule's substance is this survival across the variance-vacuous regime, not the elementary even-odd orthogonality itself.

Parity throughout this Letter refers to the parity of a steady-state functional under reversal of the drive frequency, distinct from the time-reversal parity of dynamical variables in the sense of Refs.~\cite{Lee2013,Kim2022}. A concurrent result by Knight, Kaveh, and Pruessner~\cite{Knight2026} establishes a parallel constraint on the dissipation side, treated as a companion below. The Letter sets up the driven OU model and the parity object, proves the rule and its no-merger corollary, carries it into the heavy-tailed regime in which the variance-based bound is vacuous, and proposes a falsifiable test on an electrical Brownian gyrator with circuit-level $\alpha$-stable noise.

\plsection{Setup}%
The canonical model studied throughout is the $n$-dimensional driven Ornstein-Uhlenbeck (OU) process
\begin{equation}\label{eq:OU}
dx_t \;=\; -(A - \Omega)\,x_t\,dt \;+\; dN_t,
\end{equation}
in which $A = A^\top \succ 0$ is the symmetric positive-definite relaxation matrix encoding the reversible core and $\Omega = -\Omega^\top$ is the antisymmetric rotational drive. We focus on the planar case $n=2$, in which the rotational drive takes the canonical form $\Omega = \omega J$ with $J = \begin{psmallmatrix}0 & -1 \\ 1 & 0\end{psmallmatrix}$, and the reversal $\omega \to -\omega$ is the parity operation organizing the symmetry statements that follow.

The model is not a toy. The electrical autonomous Brownian gyrator of Ref.~\cite{Chiang2017} is a direct realization of Eq.~(\ref{eq:OU}) in $n=2$, with state variables identified as the two capacitor voltages and the antisymmetric drift produced by the inter-branch coupling capacitor under a temperature gradient, and the explicit parameter dictionary linking circuit elements to $(A, \Omega, D)$ is given below. The linearity of the equation is device physics, not a linearization.

The noise driving Eq.~(\ref{eq:OU}) is treated in two incarnations. In the Gaussian incarnation, $dN_t = \sigma\,dW_t$ with $W_t$ a standard Brownian motion, and the diffusion matrix $D = \tfrac12\,\sigma\sigma^\top \succ 0$ is finite. In the heavy-tailed incarnation, $dN_t = \sigma\,dL^\gamma_t$ with $L^\gamma$ an isotropic $\alpha$-stable process of tail index $\gamma \in (0,2)$ whose increments have no finite variance. The parity selection rule is stated in the Gaussian setting and carried unchanged into the heavy-tailed regime, where its symmetry origin makes the carry-over immediate.

A distinction must be drawn at the outset between two algebraic conditions on $A$, $D$, $\Omega$ that have been routinely conflated. \emph{Detailed balance}, the reversibility criterion $B D = D B^\top$ with $B = A - \Omega$~\cite{Risken1989,Godreche2019}, fixes a single $\Omega$ for each $(A, D)$ at which the entropy production vanishes; within the $[A,D]=0$ class this is $\Omega = 0$. The condition $[A,D] = 0$ that drives the rule below, by contrast, leaves $\Omega$ and the entropy production $\Sigma$ entirely free. A reversible system inside $[A,D]=0$ is a single point; the content of the parity rule lies in the genuinely driven slice $\Omega \neq 0$ on which $\Sigma > 0$.

Within this driven slice the Gaussian incarnation of Eq.~(\ref{eq:OU}) admits a zero-mean Gaussian stationary distribution with covariance $\Sss$ solving the Lyapunov equation
\begin{equation}\label{eq:lyap}
(A-\Omega)\,\Sss \;+\; \Sss\,(A-\Omega)^\top \;=\; 2D.
\end{equation}
The propagator over a lag $t$ is $\Phi(t) = e^{-(A-\Omega)t}$, and the joint covariance of two time slices $(x_0, x_t)$ is the $2n \times 2n$ block matrix with diagonals $\Sss$ and off-diagonals $\Phi(t)\,\Sss$. Every rotation-invariant (coordinate-free) symmetric information functional, in particular the snapshot mutual information $I(x_0;x_t)$, is a scalar log-determinant or spectral functional of this joint covariance. The argument of the next section turns on how $\Sss$ and $\Phi(t)$ co-transform under $\omega \to -\omega$, not on how either transforms separately: the parity argument acts on the joint covariance, not on its marginal blocks.

\plsection{Parity selection rule}%
The rule applies under the structural alignment $[A,D] = 0$, with full isotropy $A = a I$, $D = D_0 I$ yielding strengthened versions of each statement below. It has three pieces: an even branch for the symmetric information content, an odd-rooted branch for the entropy production, and a corollary in which their orthogonality forbids any tight identity.

\plsubsection{Even branch}%
The cleanest case is the fully isotropic one, with $A = a I$ and $D = D_0 I$. Substituting the ansatz $\Sss = c I$ into Eq.~(\ref{eq:lyap}), the cross terms involving $\Omega$ cancel by antisymmetry ($\Omega + \Omega^\top = 0$), and the equation collapses to $2 a c\,I = 2 D_0 I$, so
\begin{equation}\label{eq:Sss-iso}
\Sss \;=\; \frac{D_0}{a}\,I \quad \text{for every } \Omega.
\end{equation}
The stationary covariance is independent of $\omega$ outright. With propagator $\Phi(t) = e^{-a t}\,e^{\Omega t}$ and orthogonality $e^{\Omega t}\,(e^{\Omega t})^\top = I$, the conditional covariance simplifies to $\Sss - \Phi(t)\Sss\Phi(t)^\top = (D_0/a)(1 - e^{-2 a t})\,I$, and the rotation drops out a second time. The snapshot mutual information of two time slices is therefore
\begin{equation}\label{eq:MI-iso}
I(x_0; x_t) \;=\; -\frac{n}{2}\,\log\bigl(1 - r^2\bigr), \qquad r \;=\; e^{-a t},
\end{equation}
exactly independent of $\omega$. For $n = 2$, $a = 1$, and unit lag $t = 1$,
\begin{equation}\label{eq:anchor}
\Ie \;=\; -\log\bigl(1 - e^{-2}\bigr) \;=\; 0.145413457869\ldots,
\end{equation}
the isotropic blindness anchor. The constant in Eq.~(\ref{eq:anchor}) is invariant across all rotational drives, and is reproduced to within $10^{-15}$ for $\omega \in \{0, 0.5, 1, 2, 5\}$ in direct numerical solution of the Lyapunov equation. The End Matter gives the full derivation and the numerical reproduction.

Beyond the isotropic case, the alignment $[A,D]=0$ alone fixes the parity of every coordinate-free symmetric information functional. Simultaneous diagonalization makes the reversal $\omega \to -\omega$ an orthogonal reflection that conjugates the stationary covariance and the propagator with the same matrix, so the joint covariance of $(x_0, x_t)$ transforms by an orthogonal block and every log-determinant or spectral functional of it is exactly even. When $[A,D]\neq 0$, parity-evenness breaks with violation $O(\|[A,D]\|)$. The End Matter carries out the full proof and the perturbative breaking statement.

The two statements just established differ in strength and should not be conflated. The alignment $[A,D]=0$ delivers parity but not $\omega$-independence; full isotropy is strictly stronger and pins the functional to the closed-form anchor of Eq.~(\ref{eq:anchor}). Coordinate-dependent quantities, such as the mutual information of a single fixed-axis component, are not invariant under isotropy and can depend on $\omega$ even when every rotation-invariant scalar does not.

\plsubsection{Odd branch}%
The steady-state entropy production rate of a linear NESS is~\cite{Tome2006,Godreche2019}
\begin{equation}\label{eq:EP-def}
\begin{aligned}
\Sigma &\;=\; \Tr\!\bigl[\Firr^\top D^{-1} \Firr\,\Sss\bigr], \\
\Firr &\;=\; -(A - \Omega) + D\,\Sss^{-1},
\end{aligned}
\end{equation}
with $\Firr$ the irreversible component of the drift. Under isotropy $\Sss = (D_0/a) I$, so $D \Sss^{-1} = a I$, and the relaxation pieces cancel,
\begin{equation}\label{eq:Firr-iso}
\Firr \;=\; -(a I - \Omega) + a I \;=\; \Omega.
\end{equation}
The irreversible force is precisely the antisymmetric rotational drive, and flips sign exactly under $\omega \to -\omega$. Substituting Eq.~(\ref{eq:Firr-iso}) into Eq.~(\ref{eq:EP-def}) gives
\begin{equation}\label{eq:Sigma-iso}
\Sigma \;=\; \frac{1}{a}\,\Tr[\Omega^\top \Omega] \;=\; \frac{2\,\omega^2}{a} \quad (n = 2).
\end{equation}
The closed form Eq.~(\ref{eq:Sigma-iso}) is reproduced to within $10^{-15}$ in direct numerical solution of the Lyapunov equation. When $[A,D] = 0$ in $n=2$, $\Sigma$ vanishes if and only if $\Omega = 0$ and is exactly quadratic in $\omega$,
\begin{equation}\label{eq:Sigma-aligned}
\Sigma \;=\; \frac{(\Tr D)^2}{\Tr A \cdot \det D}\,\omega^2,
\end{equation}
reducing to Eq.~(\ref{eq:Sigma-iso}) at full isotropy; the derivation is in the End Matter. When $[A,D] \neq 0$, $\Sigma$ is non-negative throughout but is generically positive at $\omega = 0$, with the misalignment itself driving a stationary current, and parity is no longer exact.

The entropy production is called \emph{odd-rooted} in two senses that coincide for this model. Its directional content is the time-antisymmetric component of the path action~\cite{Maes2020,Spinney2016}, the established sense, and it flips sign under the drive reversal $\omega \to -\omega$, the present sense. The two coincide whenever $[A,D]=0$: the antisymmetric drive $\Omega$ is the source of path-action time-asymmetry, so its sign flip under $\omega \to -\omega$ flips the time-antisymmetric component. The coincidence is explicit at isotropy where $\Firr = \Omega$, and closes the even-odd orthogonality into a selection rule. The two senses do not coincide in general, and we keep them distinct throughout. The directional content of dissipation carries the sign of $\omega$, so it is odd, while the magnitude $\Sigma$ is even and isotropically quadratic. Second-law content lives at order $\omega^2$, and symmetric-information content lives at order $\omega^0$.

\begin{figure}[t]
\centering
\includegraphics[width=\linewidth]{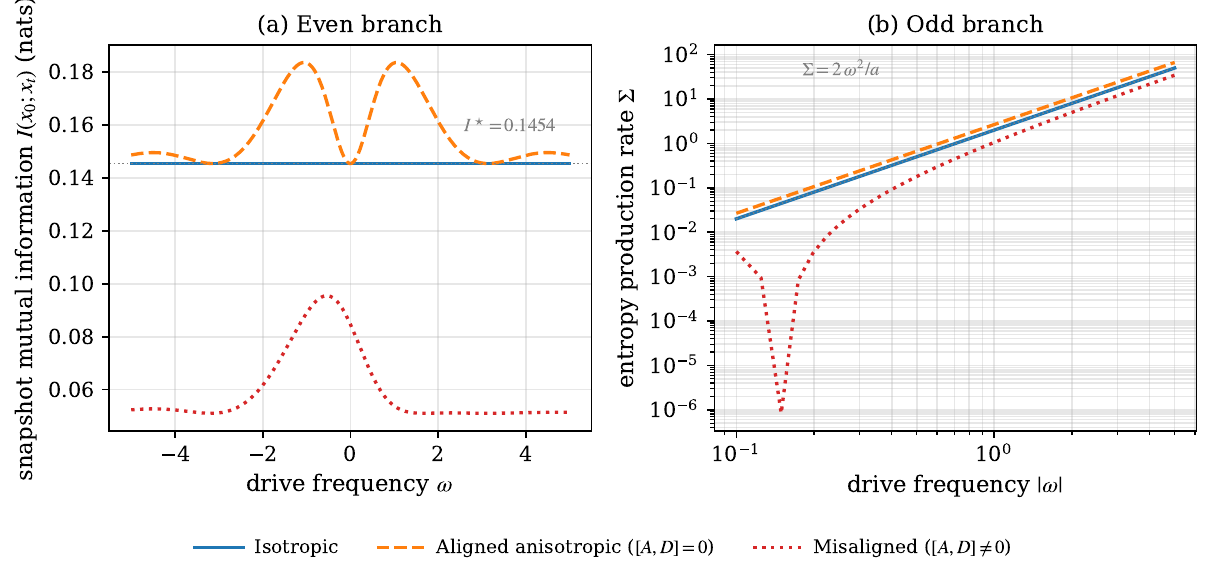}
\caption{\label{fig:even-odd}Even and odd branches at $\gamma=2$ from direct numerical solution of Eq.~(\ref{eq:lyap}), with $n=2$, $a=1$, $t=1$. Configurations: \emph{Isotropic} $A=I$, $D=I$ (blue, solid); \emph{Aligned} $A=I$, $D=\mathrm{diag}(1,3)$, $[A,D]=0$ (orange, dashed); \emph{Misaligned} $A=\mathrm{diag}(1,2)$, $D=I+0.3\,\sigma_x$ with $\sigma_x$ the off-diagonal Pauli matrix (red, dotted). (a) Snapshot MI $I(x_0;x_t)$ vs $\omega$. (b) Entropy production $\Sigma$ vs $|\omega|$, log-log; gray dotted reference $\Sigma=2\omega^2/a$.}
\end{figure}

\plsubsection{Corollary}%
The graded statements of the even branch and the parity flip of the odd branch close into a no-go. Assume $[A,D]=0$, and let $\mathcal F_{\mathrm{even}}$ be any symmetric information measure and $\mathcal F_{\mathrm{odd}}$ the antisymmetric component of the directed information rate. No non-trivial identity $\mathcal F_{\mathrm{even}}(\omega) = \mathcal F_{\mathrm{odd}}(\omega)$ can hold for every $\omega$, since parity forces
\begin{align}\label{eq:parity-chain}
\mathcal F_{\mathrm{even}}(\omega)
  &= \mathcal F_{\mathrm{even}}(-\omega) = \mathcal F_{\mathrm{odd}}(-\omega) \notag \\
  &= -\mathcal F_{\mathrm{odd}}(\omega) = -\mathcal F_{\mathrm{even}}(\omega),
\end{align}
so $\mathcal F_{\mathrm{even}} \equiv 0$, contradicting non-triviality. The tight merger is blocked by symmetry.

What parity does not forbid is a one-sided inequality. At the isotropic point both sides are closed form. The snapshot mutual information is the constant $\Ie = -(n/2)\log(1 - e^{-2 a t})$ of Eq.~(\ref{eq:MI-iso}), and the entropy production is the quadratic $\Sigma = 2\omega^2/a$ of Eq.~(\ref{eq:Sigma-iso}). Any inequality of the form $\Ie \le c\,\Sigma$ with finite $c$ holds only for $|\omega|$ above a threshold and fails as $|\omega| \to 0$, with the ratio $\Ie/\Sigma \to \infty$. The $\omega^0$-versus-$\omega^2$ parity mismatch appears explicitly as a decoupling rather than a tight law. The selection rule forbids the equality ``information $=$ dissipation'', while permitting one-sided bounds in which dissipation enters through its even square, of the kind organized by thermodynamic uncertainty relations.

FIG.~\ref{fig:even-odd} confirms the closed-form predictions of this section by direct numerical solution of the Lyapunov equation in the three configurations the rule distinguishes. The isotropic case sits on the blindness anchor $\Ie \approx 0.1454$~nats exactly across the entire drive sweep. The aligned anisotropic case with $[A,D]=0$ shows parity-evenness without blindness, the snapshot MI now an even function of $\omega$ that varies with $|\omega|$. The misaligned case with $[A,D]\neq 0$ visibly breaks parity, with the snapshot MI no longer symmetric about $\omega = 0$. The entropy production panel reproduces $\Sigma = 2\omega^2/a$ to within $10^{-15}$ in the isotropic case, with the $\Sigma \propto \omega^2$ slope persisting under $[A,D]=0$ anisotropy and breaking down under misalignment, where $\Sigma$ approaches a finite misalignment-driven floor at $\omega = 0$. FIG.~\ref{fig:nogo} converts the corollary above into two direct diagnostics: the ratio $I/\Sigma$ diverges as $|\omega| \to 0$ in both $[A,D]=0$ configurations, ruling out any inequality $\Ie \le c\,\Sigma$ with finite $c$; and the parity violation $|I(\omega) - I(-\omega)|$ sits at the floating-point floor under $[A,D]=0$, thirteen decades below the misaligned case.

\begin{figure}[!tbp]
\centering
\includegraphics[width=\linewidth]{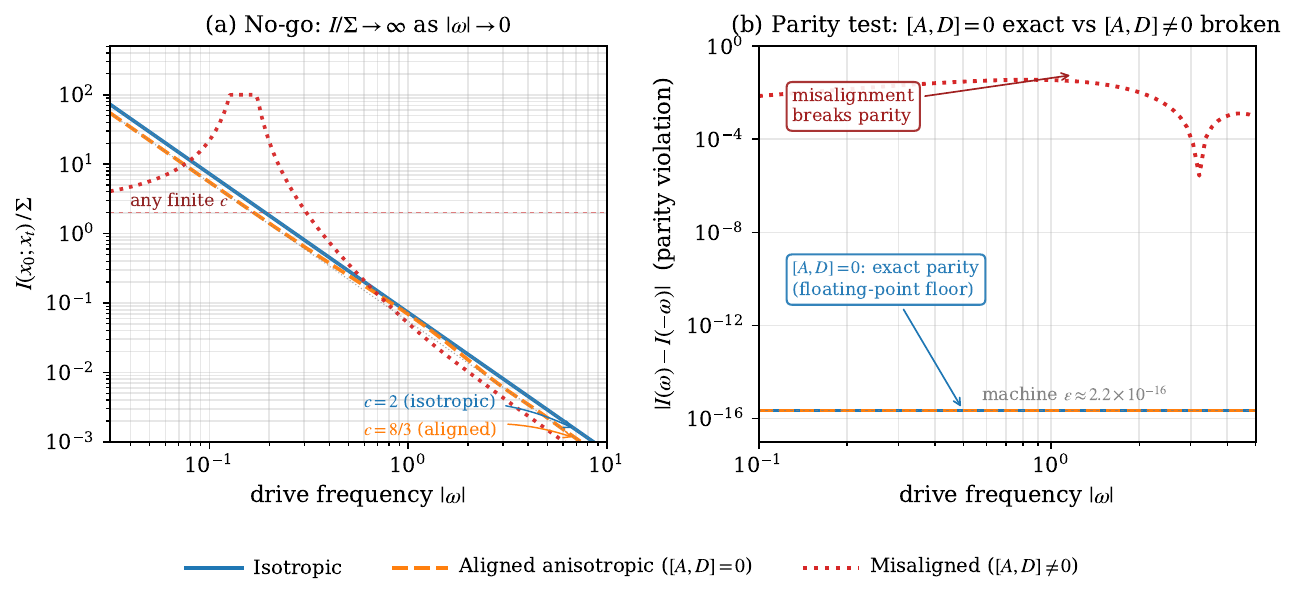}
\caption{\label{fig:nogo}No-go diagnostic, configurations as in FIG.~\ref{fig:even-odd}. (a) Ratio $I/\Sigma$ vs $|\omega|$, log-log; gray references $\Ie/(c\,\omega^2)$ at $c=2$ (isotropic) and $c=8/3$ (aligned), misaligned curve capped at $10^2$. (b) Parity violation $|I(\omega)-I(-\omega)|$ vs $|\omega|$, log-log; horizontal reference at machine $\epsilon\approx 2.2\times 10^{-16}$.}
\end{figure}

\plsection{Heavy-tailed regime}%
The even-ness established under $[A,D]=0$ was obtained from the orthogonality of the rotation $e^{\Omega t}$, and never from any moment of the noise. This is why the rule survives heavy tails.

For isotropic $\alpha$-stable noise with tail index $\gamma \in (0, 2)$, the increments of $L^\gamma_t$ are invariant in distribution under every rotation $R \in SO(n)$, by the sub-Gaussian representation $L \,\stackrel{d}{=}\, \sqrt{S}\,G$ in which $S$ is a positive $(\gamma/2)$-stable subordinator and $G$ is standard Gaussian. The rotation acts only on the Gaussian factor, and the radial part is invariant. Under full isotropy $A = a I$, $D = D_0 I$, passing to the co-rotating frame $y_t = R(-\omega t)\,x_t$ converts Eq.~(\ref{eq:OU}) with isotropic $\alpha$-stable driving into a rotation-free isotropic OU--L\'evy process for $y_t$, and any rotation-invariant symmetric-information functional of the lab-frame process equals, by construction, that of the rotation-free process, and so is independent of $\omega$ by the same symmetry argument that delivered the Gaussian case, without ever requiring $\Sss$ to be finite. The sub-Gaussian construction in the End Matter proves finiteness and $\omega$-independence of the snapshot mutual information in the heavy-tailed regime on its own footing.

The variance-based thermodynamic uncertainty relation fares differently. The relation lower-bounds the relative variance $\mathrm{Var}(J_T)/\langle J_T\rangle^2$ of any steady-state current $J_T$ by $2/\Sigma$ at long observation time. The current relevant to the parity rule is the circulation observable
\begin{equation}\label{eq:Theta}
\Theta_T \;=\; \int_0^T \bigl(x_1(s)\,dx_2(s) - x_2(s)\,dx_1(s)\bigr),
\end{equation}
a functional of the trajectory increments that inherits the heavy tail of the driving noise. For $\gamma < 2$ it has tail index $\gamma/2$ and infinite variance~\cite{Spinney2016}. The left-hand side of the variance bound is therefore $+\infty$, and the inequality is vacuous in this regime. On the dissipation side, the entropy production rate must be defined for $\gamma < 2$ through the path-measure relative entropy with the nonlocal jump kernel of the $\alpha$-stable generator, not through the Gaussian quadratic form Eq.~(\ref{eq:EP-def}). The absolute value of $\Sigma$ is not claimed to be finite without a tempering regularization, and that caveat is stated rather than hidden. What \emph{is} a symmetry statement, and is therefore cutoff-independent, is the parity structure under $[A,D] = 0$: $\Omega = 0 \Leftrightarrow \Sigma = 0$, monotonicity in $|\omega|$, and the small-$\omega$ scaling $\Sigma \propto \omega^2$.

The two separate cleanly. The selection rule remains exact precisely where the variance-based bound becomes vacuous, because the rule depends only on the rotational symmetry of the drift, while the bound depends on the second moment of a current. No $\gamma < 2$ fractional-moment numerical values are used anywhere in the analysis. The heavy-tailed claims of the present section rest only on the symmetry argument above and on the Gaussian anchors Eqs.~(\ref{eq:anchor}) and~(\ref{eq:Sigma-iso}).

\plsection{Experimental proposal}%
The most direct platform for testing the parity selection rule is the electrical autonomous Brownian gyrator of Chiang \emph{et al.}~\cite{Chiang2017}, two capacitively coupled RC branches whose resistors sit at temperatures $T_1 \neq T_2$. The temperature gradient sustains a stationary gyrating current in the $(V_1, V_2)$ plane of the two capacitor voltages, the linear-NESS realization studied since Refs.~\cite{Filliger2007,Dotsenko2013}.

The device is a literal realization of Eq.~(\ref{eq:OU}) in $n = 2$. The equations of motion are linear,
\begin{equation}\label{eq:gyrator}
\hat R\,\hat C\,\dot V \;=\; -V + \xi,
\end{equation}
with $\hat R = \diag(R_1, R_2)$ the resistance matrix and $\hat C$ the capacitance matrix carrying the inter-branch coupling $-C_c$ off-diagonally. The state variables are the two capacitor voltages $V = (V_1, V_2)^\top$; the off-diagonal coupling combined with the resistance imbalance $R_1 \neq R_2$ makes the antisymmetric part of the effective drift $\hat C^{-1}\hat R^{-1}$ nonzero, and the steady state is Gaussian with covariance solving the Lyapunov equation~(\ref{eq:lyap}). The native noise is Gaussian Johnson-Nyquist,
\begin{equation*}
\langle \xi_i(t)\,\xi_j(t')\rangle \;=\; 2\,k_B T_i R_i\,\delta_{ij}\,\delta(t - t').
\end{equation*}
Throughout this section $\Omega$ denotes the antisymmetric drive of the rule, set in practice by augmenting the base gyrator with an externally controlled non-reciprocal element (discussed below).

The gyration rate in the $(V_1, V_2)$ plane scales as $\langle \dot\phi\rangle \propto \epsilon$ with $\epsilon = 2 k_B R_1 R_2 C_c (T_2 - T_1)$, jointly controlled by the coupling capacitor $C_c$ and the temperature contrast~\cite{Chiang2017}. The entropy production is measurable redundantly, both from the trajectory work $Q_i = \oint V_i\,dq_i$ and from the closed-form resistive heat rate
\begin{equation}\label{eq:gyrator-heat}
\langle \dot Q\rangle \;=\; \frac{C_c^2\,k_B(T_2 - T_1)}{\det \hat C \cdot \mathrm{Tr}(\hat R \hat C)},
\end{equation}
the latter following from the Lyapunov solution. Both readouts have been cross-checked against experiment in Ref.~\cite{Chiang2017}, so the dissipation prediction $\Sigma = 2\omega^2/a$ of Eq.~(\ref{eq:Sigma-iso}) has an unambiguous independent measurement target.

Two layers of testable claims should be separated. The strongest claim, the exact $\omega$-blindness of Eq.~(\ref{eq:MI-iso}), is the full isotropic prediction and requires operating near the symmetric point $R_1 \approx R_2$ with the externally controlled $\Omega$ specified below. A weaker device-generic claim is the parity split itself, the prediction that small-$\omega$ mutual information stays at order $\omega^0$ while the entropy production rises as $\omega^2$, so that their ratio diverges as $|\omega| \to 0$. This decoupling survives moderate anisotropy, because the leading even and odd orders are unchanged. The experiment can therefore verify the $\omega^0$-versus-$\omega^2$ mismatch at generic device parameters and approach exact blindness by tuning toward the symmetric point.

Access to the heavy-tailed regime is a single well-defined hardware composition rather than a thought experiment, and both ingredients exist. The first is the gyrator itself~\cite{Chiang2017}, the verified 2D-OU rotational NESS. The second is a circuit-level $\alpha$-stable noise generator with tunable stability index, demonstrated previously on a FitzHugh-Nagumo circuit~\cite{LevyCircuit2024}. The composition is a programmable current source, a digitally synthesized isotropic $\alpha$-stable sequence driving a wideband transconductance stage through a digital-to-analog converter, superposed on or replacing one branch's Johnson-Nyquist source. The injection enters at the noise port, the RC network providing the linear drift $A - \Omega$ is left untouched, the 2D-OU mapping is preserved, and only the tail index $\gamma$ of $dN_t$ changes.

The proposal has one binding design constraint. In the abstract SDE Eq.~(\ref{eq:OU}), the antisymmetric drive $\Omega$ enters the drift independently of the diffusion matrix $D$. In the thermal gyrator of Ref.~\cite{Chiang2017} without the augmentation, however, gyration is \emph{induced} by the bath asymmetry $T_1 \neq T_2$, so the bath-induced effective antisymmetric drift and the anisotropy of $D$ are kinematically coupled through the gyration-rate scaling $\epsilon = 2 k_B R_1 R_2 C_c (T_2 - T_1)$. Injecting an isotropic $\alpha$-stable noise of equal strength on both branches removes the very temperature asymmetry that drives gyration, sending this effective drift to zero together with $\Sigma \to 0$, and reduces the heavy-tailed experiment to an empty test. A clean realization of the rule's $\Omega$ knob---required for both the Gaussian blindness anchor and the heavy-tailed extension---therefore needs $\Omega$ to enter the drift through a mechanism independent of the bath, for instance via a nonreciprocal coupling element implementing $\Omega$ directly (a literal gyrator element in the sense of Tellegen~\cite{Tellegen1948}), or via active electronic feedback. Once $\Omega$ is sourced independently of $D$, isotropic L\'evy injection on both branches probes the parity rule in its pure form, with $\gamma$ tuned and $\Omega$ held fixed.

Three predictions are then sharp and falsifiable. The mutual information between the two capacitor voltages is even in $\omega$, and at the symmetric point exactly independent of $\omega$, with the Gaussian closed form of Eq.~(\ref{eq:MI-iso}) at the device's $at$. The entropy production is exactly $\propto \omega^2$ at the symmetric point and remains $O(\omega^2)$ generically near $|\omega| = 0$, with the prefactor of Eq.~(\ref{eq:Sigma-iso}). In the $\gamma < 2$ heavy-tailed regime the variance-based thermodynamic uncertainty relation breaks down, while the parity split between mutual information and entropy production is unmoved. Measurability is anchored independently. Mutual information and transfer entropy are routinely measured in driven colloidal nonequilibrium steady states~\cite{Colloidal2025}, and dissipation has been measured in the gyrator family at the level of single-trajectory work~\cite{Chiang2017,Filliger2007}.

\plsection{Discussion and conclusion}%
The identification of the directed information rate with the entropy production on this NESS is textbook stochastic thermodynamics~\cite{Horowitz2014,Ito2013,Sagawa2013,Spinney2016,Maes2020}. That odd half is not new in the present Letter. What is added is the structural complement: the even half stated as an information-theoretic proposition, that the symmetric information content is parity-even and isotropically blind to the drive; the packaging of the even and odd halves into a single selection rule whose orthogonality forbids the tight merger of information with the second law; the diagnosis, at the level of parities, of the structural obstruction to equality-type information thermodynamics, reducing to the $\omega^0$-versus-$\omega^2$ parity mismatch; and the extension of the rule into the heavy-tailed regime in which the variance bound is vacuous, by an argument that uses only symmetry and never any moment.

The selection rule does not contest the inequality-type information thermodynamics of Refs.~\cite{Horowitz2014,Ito2013,Sagawa2013,Seifert2012} in which entropy production bounds an information change from one side. Those inequalities are untouched and are in fact consistent with the loose one-sided bound permitted by the rule. What symmetry forbids is the tight, equality-type merger, not the inequality.

The Letter's contribution separates cleanly from several concurrent and adjacent threads. The work of Ref.~\cite{Entropy2024NEEnhance}, sometimes summarized as ``information enhanced by driving'', measures a different object, a nonequilibrium drive that reshapes the stationary law, whereas the divergence-free rotation of the present setup leaves the stationary law invariant by Eq.~(\ref{eq:Sss-iso}). The inverse-TUR no-go of Ref.~\cite{Vo2025} is not triggered, because the variance divergence is induced by the input statistics rather than by a vanishing relaxation gap, and the parity structure of the entropy production is cutoff-free. The entropic uncertainty relation of Ref.~\cite{Hasegawa2025}, the coherent-transport TUR of Ref.~\cite{BrandnerSaito2025}, and the mutual-information-rate decompositions of Refs.~\cite{Cho2025,Zeng2017} all operate on the odd branch of the information-entropy decomposition, and none of them states the even and blind half as an information-theoretic proposition or packages the parity rule. Information-geometric reformulations of NESS transitions~\cite{Ito2020,Lacerda2025} treat dissipation through Riemannian geodesics in parameter space, a direction adjacent to but distinct from the symmetry-based analysis here.

A concurrent result on the same methodological footing has recently appeared. Knight, Kaveh, and Pruessner~\cite{Knight2026} establish that the parity and time-reversal symmetries of a hidden self-propulsion variable fix the partial entropy production of active particles~\cite{Fodor2016}, with non-trivial entropy production entering only at sixth order in the self-propulsion speed. Their result and the parity selection rule are both instances of symmetry-based constraints on nonequilibrium steady-state functionals, on the dissipation side and on the symmetric information side, respectively, and are read here as companion statements rather than as competitors.

\plsection{Acknowledgments}%
This work was supported in part by the National Natural Science Foundation of China under Grants 62571450 and 62101450; the Key Research and Development Program of Shaanxi under Grants 2025CY-YBXM-043 and 2025CG-GJHX-15; the Shanghai Academy of Spaceflight Technology under Grant SAST2025-037; the Open Fund of Intelligent Control Laboratory; the Open Fund of Key Laboratory of Radio Spectrum Testing Technology (The State Radio Monitoring Center Testing Center), Ministry of Industry and Information Technology.

\bibliographystyle{apsrev4-2}
\bibliography{refs}


\clearpage
\onecolumngrid

\begin{center}
\textbf{\large END MATTER}
\end{center}

\vspace{1em}

\setcounter{section}{0}
\renewcommand{\thesection}{S\arabic{section}}
\setcounter{equation}{0}
\renewcommand{\theequation}{S\arabic{equation}}
\setcounter{figure}{0}
\renewcommand{\thefigure}{S\arabic{figure}}

\noindent This End Matter gives complete proofs of every analytic statement made in the main text, and lays out the foundational lemmas the proofs rely on. The notation is identical to the main text. The model is the $n$-dimensional driven Ornstein--Uhlenbeck (OU) process
\begin{equation}\label{eq:S-OU}
dx_t \;=\; -(A - \Omega)\,x_t\,dt \;+\; dN_t,
\end{equation}
with $A = A^\top \succ 0$ symmetric positive definite, $\Omega = -\Omega^\top$ antisymmetric, and $N_t$ a noise process with covariance generator $2D$, $D = D^\top \succ 0$. The Gaussian-noise stationary covariance $\Sss$ obeys the Lyapunov equation
\begin{equation}\label{eq:S-lyap}
(A - \Omega)\,\Sss \;+\; \Sss\,(A - \Omega)^\top \;=\; 2 D,
\end{equation}
with propagator $\Phi(t) = e^{-(A - \Omega) t}$. We focus on the planar case $n = 2$ with $\Omega = \omega J$ and $J = \begin{psmallmatrix}0 & -1 \\ 1 & 0\end{psmallmatrix}$. The steady-state entropy production rate of the Gaussian linear NESS is
\begin{equation}\label{eq:S-EP-def}
\Sigma \;=\; \Tr\!\bigl[\Firr^\top D^{-1} \Firr\,\Sss\bigr], \qquad
\Firr \;=\; -(A - \Omega) + D\,\Sss^{-1}.
\end{equation}

\section{Even-branch proof under $[A,D]=0$}\label{app:even}

\subsection{Foundational lemmas}\label{app:prelim}

The proofs in this and the following appendices rely on four foundational facts that we record here. Three are stated immediately; the fourth (the trace identity) appears in Sec.~\ref{app:trace-id} where it is first used.

\begin{lemma}[Stability and Lyapunov existence]\label{lem:stability}
Let $A \in \RR^{n \times n}$ be symmetric positive definite and $\Omega \in \RR^{n \times n}$ antisymmetric. Then every eigenvalue of $A - \Omega$ has strictly positive real part, and the Lyapunov equation
\begin{equation}\label{eq:S-lyap-general}
(A - \Omega)\,\Sss + \Sss\,(A - \Omega)^\top = 2 D
\end{equation}
admits a unique positive-definite solution $\Sss$ for every symmetric positive-definite $D$. Explicitly,
\begin{equation}\label{eq:S-lyap-integral}
\Sss = 2 \int_0^\infty e^{-(A - \Omega) s}\, D\, e^{-(A - \Omega)^\top s}\, ds.
\end{equation}
\end{lemma}
\begin{proof}
Let $\lambda \in \mathbb{C}$ be an eigenvalue of $A - \Omega$ with right eigenvector $v \in \mathbb{C}^n \setminus \{0\}$. Then
\[
v^* (A - \Omega) v \;=\; v^* A v \;-\; v^* \Omega v \;=\; \lambda\,v^* v.
\]
Since $A \succ 0$ is real symmetric, $v^* A v \in \RR_{>0}$. Since $\Omega$ is real antisymmetric, $v^* \Omega v \in i\RR$. Hence
\[
\Re \lambda \;=\; \frac{v^* A v}{v^* v} \;>\; 0,
\]
proving stability. The Sylvester operator $X \mapsto (A - \Omega) X + X (A - \Omega)^\top$ is invertible because the spectra of $A - \Omega$ and $-(A - \Omega)^\top$ are disjoint (one in the open right half-plane, the other in the open left half-plane). The integral~\eqref{eq:S-lyap-integral} converges by stability of $A - \Omega$, satisfies the Lyapunov equation by direct differentiation under the integral, and is the unique solution. It is positive definite because $D \succ 0$ and the conjugating exponentials are nonsingular.
\end{proof}

\begin{lemma}[Reflection generating the parity flip in $n=2$]\label{lem:reflection}
Let $J = \begin{psmallmatrix}0 & -1 \\ 1 & 0\end{psmallmatrix}$ and $R = \diag(1, -1)$. Then $R$ is involutory ($R^2 = I$), symmetric ($R^\top = R$), and orthogonal with $\det R = -1$. Moreover:
\begin{enumerate}
\item[(i)] $R J R = -J$;
\item[(ii)] For any diagonal $A = \diag(a_1, a_2)$ and $D = \diag(d_1, d_2)$, $R A R = A$ and $R D R = D$.
\end{enumerate}
\end{lemma}
\begin{proof}
$R^2 = I$ and $R^\top = R$ are immediate. For (i),
\[
R J \;=\; \begin{pmatrix} 1 & 0 \\ 0 & -1 \end{pmatrix}\!\begin{pmatrix} 0 & -1 \\ 1 & 0 \end{pmatrix} \;=\; \begin{pmatrix} 0 & -1 \\ -1 & 0 \end{pmatrix},
\quad
(R J) R \;=\; \begin{pmatrix} 0 & -1 \\ -1 & 0 \end{pmatrix}\!\begin{pmatrix} 1 & 0 \\ 0 & -1 \end{pmatrix} \;=\; \begin{pmatrix} 0 & 1 \\ -1 & 0 \end{pmatrix} \;=\; -J.
\]
For (ii), conjugation by the sign flip on axis $2$ leaves any diagonal matrix invariant.
\end{proof}

\begin{remark}\label{rem:rotations-do-not-reverse}
Rotations $R(\theta) \in SO(2)$ do not reverse $J$. Since $J$ generates $SO(2)$, $R(\theta) J R(\theta)^\top = J$ for every $\theta$. The parity flip $\Omega \to -\Omega$ is achieved only by orientation-reversing elements of $O(2)$, of which $R = \diag(1,-1)$ is the canonical representative. This is the geometric reason the parity selection rule is a discrete-symmetry statement, not a continuous one.
\end{remark}

The mutual-information formula below is standard~\cite{Shannon1948,CoverThomas2006}.

\begin{lemma}[Gaussian mutual information]\label{lem:gaussian-MI}
Let $(X, Y) \in \RR^n \times \RR^n$ be jointly Gaussian with zero mean and block covariance
\[
\mathcal C \;=\; \begin{pmatrix} \Sigma_X & \Sigma_{XY} \\ \Sigma_{YX} & \Sigma_Y \end{pmatrix},
\]
with $\Sigma_X, \Sigma_Y \succ 0$. Then
\begin{equation}\label{eq:S-MI-logdet}
I(X; Y) \;=\; \tfrac{1}{2}\,\log\!\frac{\det \Sigma_X \cdot \det \Sigma_Y}{\det \mathcal C} \quad \text{(nats)}.
\end{equation}
\end{lemma}
\begin{proof}
By the Gaussian differential-entropy formula,
\[
h(X) \;=\; \tfrac{1}{2}\log\bigl[(2\pi e)^n \det \Sigma_X\bigr], \quad
h(Y) \;=\; \tfrac{1}{2}\log\bigl[(2\pi e)^n \det \Sigma_Y\bigr], \quad
h(X,Y) \;=\; \tfrac{1}{2}\log\bigl[(2\pi e)^{2n} \det \mathcal C\bigr].
\]
Then $I(X;Y) = h(X) + h(Y) - h(X,Y)$ collapses to Eq.~\eqref{eq:S-MI-logdet}.
\end{proof}

\subsection{Co-reflection identities}\label{app:co-reflect}

We work in the planar case $n=2$. The structural alignment $[A,D] = 0$ means $A$ and $D$ are simultaneously diagonalizable: there exists $U \in O(2)$ such that
\[
U^\top A\,U \;=\; \diag(a_1, a_2), \qquad U^\top D\,U \;=\; \diag(d_1, d_2).
\]
In the rotated basis $\tilde x = U^\top x$, the dynamics read
\[
d\tilde x \;=\; -(U^\top A U \;-\; U^\top \Omega U)\,\tilde x\,dt \;+\; U^\top dN_t,
\]
and the rotation-invariance of the noise statistics, a property shared by Brownian and isotropic stable processes alike, returns an OU-type SDE with diagonal $A$, diagonal $D$, and $\tilde \Omega = U^\top \Omega U$.

The reversal $\omega \to -\omega$ corresponds to $\Omega \to -\Omega$. In the rotated basis, where $A$ and $D$ are diagonal, the antisymmetric matrix $-\Omega$ is obtained from $\Omega$ by conjugation with the reflection $R = \diag(1, -1)$ of Lemma~\ref{lem:reflection}, which preserves the diagonal forms of $A$ and $D$ and reverses the orientation in the plane of $\Omega$. The propagator $\Phi(t) = e^{-(A - \Omega) t}$ then satisfies
\begin{equation}\label{eq:S-Phi-conjugate}
\Phi_{-\omega}(t) \;=\; R\,\Phi_{+\omega}(t)\,R^\top,
\end{equation}
since $R(A - \Omega)R^\top = A + \Omega$ and the matrix exponential is preserved under orthogonal conjugation.

Substituting $\Omega \to -\Omega$ into the Lyapunov equation~\eqref{eq:S-lyap}, multiplying on the left by $R$ and on the right by $R^\top$, and using $R A R^\top = A$, $R D R^\top = D$, $R \Omega R^\top = -\Omega$ (Lemma~\ref{lem:reflection}) shows that $\Sss(-\omega)$ obeys the same Lyapunov equation as $R\,\Sss(+\omega)\,R^\top$. Uniqueness of the positive-definite Lyapunov solution (Lemma~\ref{lem:stability}) gives
\begin{equation}\label{eq:S-Sss-conjugate}
\Sss(-\omega) \;=\; R\,\Sss(+\omega)\,R^\top.
\end{equation}
Together, Eqs.~\eqref{eq:S-Phi-conjugate} and~\eqref{eq:S-Sss-conjugate} are the co-reflection identities.

The joint covariance of $(x_0, x_t)$ is the $2n \times 2n$ block matrix
\[
\mathcal C(\omega) \;=\; \begin{pmatrix} \Sss(\omega) & \Phi(t)\,\Sss(\omega) \\ \Sss(\omega)\,\Phi(t)^\top & \Sss(\omega) \end{pmatrix}.
\]
Under $\omega \to -\omega$ each block transforms by $R$ on the left and $R^\top$ on the right, so
\begin{equation}\label{eq:S-jointcov-transform}
\mathcal C(-\omega) \;=\; (R \oplus R)\,\mathcal C(+\omega)\,(R \oplus R)^\top.
\end{equation}
Every spectral or log-determinant functional of $\mathcal C$ is invariant under such orthogonal conjugation, so $\mathcal F(\mathcal C(-\omega)) = \mathcal F(\mathcal C(+\omega))$ for every coordinate-free symmetric information functional $\mathcal F$, proving the parity statement of the even branch in the main text.

\subsection{Isotropic blindness anchor}\label{app:isotropy-anchor}

For the strictly stronger isotropic blindness, $A = a I$, $D = D_0 I$. Substituting the ansatz $\Sss = c I$ into Eq.~\eqref{eq:S-lyap}, the cross terms involving $\Omega$ cancel by antisymmetry ($\Omega + \Omega^\top = 0$), leaving $2 a c I = 2 D_0 I$, hence
\begin{equation}\label{eq:S-Sss-iso}
\Sss \;=\; \frac{D_0}{a}\,I \quad \text{for every } \Omega.
\end{equation}
With $\Phi(t) = e^{-a t}\,e^{\Omega t}$ and $e^{\Omega t}\,(e^{\Omega t})^\top = I$, the joint covariance is
\[
\mathcal C(\omega) \;=\; \frac{D_0}{a}\,\begin{pmatrix} I & e^{-a t}\,e^{\Omega t} \\ e^{-a t}\,e^{-\Omega t} & I \end{pmatrix}.
\]
The eigenvalues of $\mathcal C$ are $(D_0/a)(1 \pm e^{-a t})$, each with multiplicity $n$, and are independent of $\omega$. By Lemma~\ref{lem:gaussian-MI},
\[
I(x_0; x_t) \;=\; \tfrac{1}{2}\,\log\!\frac{\det \Sss \cdot \det \Sss}{\det \mathcal C} \;=\; -\tfrac{n}{2}\,\log\bigl(1 - e^{-2 a t}\bigr),
\]
which is exact $\omega$-independence at the closed-form value
\begin{equation}\label{eq:S-MI-iso}
I(x_0; x_t) \;=\; -\tfrac{n}{2}\,\log\bigl(1 - e^{-2 a t}\bigr).
\end{equation}
At $n = 2$, $a = 1$, $t = 1$,
\begin{equation}\label{eq:S-anchor}
\Ie \;=\; -\log\bigl(1 - e^{-2}\bigr) \;=\; 0.14541345786\ldots
\end{equation}
nats. Direct numerical solution of the Lyapunov equation followed by evaluation of the Gaussian mutual-information formula reproduces this value across the sweep $\omega \in \{0, 0.5, 1, 2, 5\}$ with variation $\le 10^{-15}$, consistent with floating-point round-off.

\subsection{Failure modes}\label{app:fail}

When $[A,D] \neq 0$, no single orthogonal $R$ achieves $R A R^\top = A$, $R D R^\top = D$, and $R \Omega R^\top = -\Omega$ simultaneously, so the identities~\eqref{eq:S-Phi-conjugate} and~\eqref{eq:S-Sss-conjugate} fail. Sec.~\ref{app:perturb} below shows that the residual scales linearly in $\|[A,D]\|$ for small misalignment, which is the sharp two-sided characterization referred to in the main text.

\section{Odd-branch identity, blindness anchor numerics, and heavy-tailed proof}\label{app:odd}

\subsection{Trace identity}\label{app:trace-id}

\begin{lemma}[Trace identity for OU steady states]\label{lem:trace-id}
For any OU process satisfying the hypotheses of Lemma~\ref{lem:stability}, the stationary covariance obeys
\begin{equation}\label{eq:S-trace-id}
\Tr[D\,\Sss^{-1}] \;=\; \Tr A.
\end{equation}
\end{lemma}
\begin{proof}
The Lyapunov equation~\eqref{eq:S-lyap-general} reads
\[
A \Sss + \Sss A - \Omega \Sss + \Sss \Omega \;=\; 2 D
\]
(using $A^\top = A$ and $\Omega^\top = -\Omega$). Right-multiplying by $\Sss^{-1}$ and taking the trace,
\[
\Tr A + \Tr[\Sss A \Sss^{-1}] + \Tr[\Sss \Omega \Sss^{-1}] - \Tr \Omega \;=\; 2 \Tr[D \Sss^{-1}].
\]
By cyclic invariance of the trace, $\Tr[\Sss A \Sss^{-1}] = \Tr A$ and $\Tr[\Sss \Omega \Sss^{-1}] = \Tr \Omega$; and $\Tr \Omega = 0$ since $\Omega$ is antisymmetric. The identity~\eqref{eq:S-trace-id} follows.
\end{proof}

\noindent As a corollary, $\Tr \Firr = -\Tr(A - \Omega) + \Tr[D\Sss^{-1}] = -\Tr A + \Tr A = 0$, the volume-preserving form of the irreversible OU current.

\subsection{Closed form for the entropy production at the isotropic point}\label{app:Sigma-iso}

Under isotropy $\Sss = (D_0/a) I$, so $D\,\Sss^{-1} = a I$, and the irreversible force in Eq.~\eqref{eq:S-EP-def} reduces to
\[
\Firr \;=\; -(a I - \Omega) + a I \;=\; \Omega.
\]
Substituting into the entropy production rate Eq.~\eqref{eq:S-EP-def} and using $\Omega^\top = -\Omega$, $\Sss = (D_0/a) I$, $D^{-1} = (1/D_0) I$,
\[
\Sigma \;=\; \Tr\!\bigl[\Omega^\top D^{-1} \Omega\,\Sss\bigr] \;=\; \Tr\!\Bigl[\Omega^\top \cdot \tfrac{1}{D_0}\Omega \cdot \tfrac{D_0}{a} I\Bigr] \;=\; \frac{1}{a}\,\Tr\!\bigl[\Omega^\top \Omega\bigr].
\]
In the planar case $n = 2$ with $\Omega = \omega J$ and $J^\top J = I$, $\Tr[\Omega^\top \Omega] = \omega^2 \Tr[J^\top J] = 2\omega^2$, so
\begin{equation}\label{eq:S-Sigma-iso}
\Sigma \;=\; \frac{2\,\omega^2}{a}.
\end{equation}
The result rests on $A$ and $D$ being scalar multiples of the identity.

\subsection{Closed form under $[A,D]=0$ in $n=2$}\label{app:sigma-aligned}

Beyond the isotropic point, the structural alignment $[A,D]=0$ admits a closed form for $\Sigma$ that is exactly quadratic in $\omega$ for every aligned $(A,D)$ in $n=2$. In the simultaneous eigenbasis of $A$ and $D$, write
\[
A = \diag(a_1, a_2), \qquad D = \diag(d_1, d_2),
\]
with $\Omega = \omega J$ and $J$ as in Lemma~\ref{lem:reflection}. Define the symmetric functions
\begin{equation}\label{eq:S-aligned-symfuncs}
\sigma = a_1 + a_2 = \Tr A, \quad \pi = a_1 a_2 = \det A, \quad \delta = d_1 + d_2 = \Tr D, \quad \mu = d_1 d_2 = \det D, \quad \tau = d_1 a_2 - d_2 a_1.
\end{equation}

\paragraph{Lyapunov solution.}
Writing $\Sss = (s_{ij})$ with $s_{12} = s_{21}$, the three independent entries of the Lyapunov equation $(A - \Omega) \Sss + \Sss (A + \Omega) = 2D$ give the linear system
\begin{subequations}\label{eq:S-lyap-components}
\begin{align}
a_1 s_{11} + \omega s_{12} &\;=\; d_1, \label{eq:S-lyap-11}\\
(a_1 + a_2) s_{12} + \omega(s_{22} - s_{11}) &\;=\; 0, \label{eq:S-lyap-12}\\
- \omega s_{12} + a_2 s_{22} &\;=\; d_2. \label{eq:S-lyap-22}
\end{align}
\end{subequations}
Solving Eq.~\eqref{eq:S-lyap-11} and~\eqref{eq:S-lyap-22} for $s_{11}$ and $s_{22}$ in terms of $s_{12}$, substituting into Eq.~\eqref{eq:S-lyap-12}, and using $a_1 + a_2 = \sigma$, $a_1 a_2 = \pi$, $a_2 d_1 - a_1 d_2 = -\tau$ yields
\[
\sigma(\pi + \omega^2)\,s_{12} \;=\; \omega\,\tau,
\]
hence
\begin{equation}\label{eq:S-Sss-aligned}
s_{12} \;=\; \frac{\tau\,\omega}{\sigma\,(\pi+\omega^2)}, \qquad
s_{ii} \;=\; \frac{d_i\,\sigma\,\pi + a_i\,\delta\,\omega^2}{a_i\,\sigma\,(\pi+\omega^2)}, \quad i = 1, 2.
\end{equation}
The polynomial factor $\pi + \omega^2$ in every denominator is the structural feature that drives the cancellation below.

\paragraph{Reduction of $\Sigma$ via the trace identity.}
By Lemma~\ref{lem:trace-id}, $\Tr[D \Sss^{-1}] = \sigma$. Writing $\Firr = -B + D\Sss^{-1}$ with $B = A - \Omega$ and expanding Eq.~\eqref{eq:S-EP-def},
\[
\Sigma \;=\; \Tr[D \Sss^{-1}] - 2 \Tr B + \Tr[B^\top D^{-1} B \Sss] \;=\; -\sigma + \Tr[B^\top D^{-1} B \Sss],
\]
where $\Tr B = \Tr A - \Tr \Omega = \sigma$ since $\Omega$ is antisymmetric.

\paragraph{Computation of $\Tr[B^\top D^{-1} B \Sss]$.}
Expanding
\[
B^\top D^{-1} B \;=\; (A + \Omega) D^{-1} (A - \Omega) \;=\; A D^{-1} A - A D^{-1} \Omega + \Omega D^{-1} A - \Omega D^{-1} \Omega
\]
and tracing against $\Sss$ component-wise, with $D, A$ diagonal and $\Omega = \omega J$,
\begin{align*}
\Tr[A D^{-1} A \Sss] &\;=\; \frac{a_1^2}{d_1}\,s_{11} + \frac{a_2^2}{d_2}\,s_{22}, \\
\Tr[\Omega D^{-1} \Omega \Sss] &\;=\; -\omega^2 \Bigl(\frac{s_{11}}{d_2} + \frac{s_{22}}{d_1}\Bigr), \\
-\Tr[A D^{-1} \Omega \Sss] + \Tr[\Omega D^{-1} A \Sss] &\;=\; -2 \omega s_{12} \,\frac{\tau}{\mu} \;=\; -\frac{2 \omega^2 \tau^2}{\mu\,\sigma\,(\pi + \omega^2)},
\end{align*}
where the cross terms use $a_2/d_2 - a_1/d_1 = \tau/\mu$.

\paragraph{Algebraic cancellation.}
Substituting Eq.~\eqref{eq:S-Sss-aligned} and collecting over the common denominator $\mu \sigma (\pi + \omega^2)$,
\[
\Tr[B^\top D^{-1} B \Sss] \;=\; \frac{\sigma^2 \pi \mu + N_2\,\omega^2 + \delta^2 \omega^4}{\mu \sigma (\pi + \omega^2)},
\]
with $N_2 = \delta(a_1^2 d_2 + a_2^2 d_1) - 2\tau^2 + \sigma(a_2 d_1^2 + a_1 d_2^2)$. Expanding via
\begin{gather*}
\delta(a_1^2 d_2 + a_2^2 d_1) \;=\; \mu(a_1^2 + a_2^2) + a_1^2 d_2^2 + a_2^2 d_1^2, \\
\sigma(a_2 d_1^2 + a_1 d_2^2) \;=\; \pi(d_1^2 + d_2^2) + a_1^2 d_2^2 + a_2^2 d_1^2, \\
\tau^2 \;=\; a_2^2 d_1^2 - 2\mu\pi + a_1^2 d_2^2, \quad
\sigma^2 \;=\; (a_1^2 + a_2^2) + 2\pi, \quad
\delta^2 \;=\; (d_1^2 + d_2^2) + 2\mu,
\end{gather*}
the four-term sum simplifies to $N_2 = \delta^2 \pi + \mu \sigma^2$. Hence
\[
\Tr[B^\top D^{-1} B \Sss] \;=\; \frac{\sigma^2 \pi \mu + (\delta^2 \pi + \mu \sigma^2)\omega^2 + \delta^2 \omega^4}{\mu \sigma (\pi + \omega^2)} \;=\; \frac{(\pi + \omega^2)(\sigma^2 \mu + \delta^2 \omega^2)}{\mu \sigma (\pi + \omega^2)} \;=\; \sigma + \frac{\delta^2 \omega^2}{\mu \sigma}.
\]
The factor $\pi + \omega^2$ cancels exactly between numerator and denominator. Combining with $\Sigma = -\sigma + \Tr[B^\top D^{-1} B \Sss]$,
\begin{equation}\label{eq:S-Sigma-aligned}
\Sigma \;=\; \frac{\delta^2}{\mu \sigma}\,\omega^2 \;=\; \frac{(\Tr D)^2}{\det D \cdot \Tr A}\,\omega^2,
\end{equation}
exactly quadratic with no $\omega^4$ or higher correction. At full isotropy $a_1 = a_2 = a$, $d_1 = d_2 = D_0$, $\sigma = 2a$, $\mu = D_0^2$, $\delta = 2 D_0$, the prefactor evaluates to $(2 D_0)^2/(D_0^2 \cdot 2a) = 2/a$, recovering Eq.~\eqref{eq:S-Sigma-iso}. Eq.~\eqref{eq:S-Sigma-aligned} is reproduced to within $10^{-15}$ in direct Lyapunov computation across five independent aligned $(A, D)$ configurations.

\subsection{Heavy-tailed proof via the sub-Gaussian construction}\label{app:heavy-tailed}

The next four steps promote the heavy-tailed heuristic of the main text to a self-contained argument that the snapshot mutual information $I(x_0; x_t)$ is well-defined and finite for isotropic $\alpha$-stable driving with $\gamma < 2$, and is exactly independent of the drive $\omega$. The finiteness rests only on finite \emph{logarithmic} moments of the stable subordinator, and never on a finite second moment.

\paragraph{Step (a). Rotational invariance of the increments.}
An isotropic symmetric $\alpha$-stable vector $L \in \RR^n$ of stability index $\gamma$ admits the \emph{sub-Gaussian} representation~\cite{SamorodnitskyTaqqu1994}
\begin{equation}\label{eq:S-subGauss}
L \;\stackrel{d}{=}\; \sqrt{S}\,G,
\end{equation}
with $G \sim \mathcal N(0, I)$ a standard Gaussian on $\RR^n$ and $S \ge 0$ a positive $(\gamma/2)$-stable subordinator independent of $G$. For any fixed $R \in SO(n)$ one has $R G \stackrel{d}{=} G$ by rotational invariance of the standard Gaussian, while the scalar $S$ is unchanged, so $R L \stackrel{d}{=} L$. The increments of the corresponding L\'evy process are invariant in distribution under every $R \in SO(n)$, with the rotation acting only on the Gaussian factor.

\paragraph{Step (b). Conditioning on the subordinator yields Gaussian OU with log-determinant MI.}
Write $\mathcal S$ for the realization of the driving subordinator path on $[0, t]$. Conditional on $\mathcal S$, the noise is Gaussian, so the solution
\[
x_t \;=\; \Phi(t)\,x_0 \;+\; \int_0^t \Phi(t - s)\,dN_s
\]
makes $(x_0, x_t)$ jointly Gaussian. Conditioning additionally on the subordinator value $S_0$ that drives the stationary $x_0$, the conditional snapshot mutual information is the Gaussian log-determinant functional
\begin{equation}\label{eq:S-cond-MI}
I^{G}(\mathcal S, S_0) \;=\; -\tfrac{1}{2}\,\log\det\bigl(I - \Sigma_{x_0 x_t}\,\Sigma_{x_t}^{-1}\,\Sigma_{x_t x_0}\,\Sigma_{x_0}^{-1}\bigr) \;<\; \infty,
\end{equation}
finite for every realization. In the isotropic case Eq.~\eqref{eq:S-cond-MI} reduces to
\[
I^{G} \;=\; \tfrac{n}{2}\,\log(1 + r^2\,\rho),
\]
with $r = e^{-a t}$ the subordinator-independent correlation factor and $\rho = S_0 / S_{[0,t]}$ a positive scale ratio. The heavy tail enters only through $\rho$, never through $r$.

\paragraph{Step (c). Expectation over the subordinator: finiteness from log-moments, not from $\Sss$.}
The transition subordinator $\mathcal S$ on $[0, t]$ is independent of the stationary pair $(x_0, S_0)$, so $I\bigl((x_0, S_0);\,\mathcal S\bigr) = 0$ and in particular $I(x_0; \mathcal S \mid S_0) = 0$. Data processing followed by the chain rule then gives
\begin{equation}\label{eq:S-finiteness-chain}
I(x_0; x_t) \;\le\; I\bigl(x_0;\, (x_t, S_0, \mathcal S)\bigr) \;=\; I(x_0; S_0) \;+\; \E_{\mathcal S, S_0}\bigl[\,I^{G}\,\bigr],
\end{equation}
and we bound the two summands separately.

The conditional term uses the elementary inequality $\log(1 + x) \le \log 2 + |\log x|$ valid for all $x > 0$ --- for $x \ge 1$, $\log(1+x) \le \log(2x) = \log 2 + \log x = \log 2 + |\log x|$; for $0 < x < 1$, $\log(1+x) \le \log 2 \le \log 2 + |\log x|$ --- applied to $x = r^2 \rho$:
\begin{align*}
\E[I^{G}] \;&\le\; \tfrac{n}{2}\,\E[\,\log(1 + r^2 \rho)\,] \\
&\le\; \tfrac{n}{2}\Bigl(\log 2 + 2\,\log\tfrac{1}{r} + \E|\log S_0| + \E|\log S_{[0,t]}|\Bigr) \;<\; \infty,
\end{align*}
because a positive $(\gamma/2)$-stable subordinator has finite logarithmic moments, $\E|\log S| < \infty$, for every $\gamma \in (0, 2)$.

The scale term $I(x_0; S_0) = h(x_0) - h(x_0 \mid S_0)$ is finite by the same property. Under the sub-Gaussian representation, $x_0 \mid S_0 \sim \mathcal N(0, S_0\,\Sigma_0)$ with a fixed $\Sigma_0 \succ 0$, so
\[
h(x_0 \mid S_0) \;=\; \tfrac{n}{2}\log(2\pi e) \,+\, \tfrac{1}{2}\log\det\Sigma_0 \,+\, \tfrac{n}{2}\,\E[\log S_0] \;<\; \infty,
\]
and the marginal density $p_{x_0}$, a sub-Gaussian mixture, is bounded on compacts with isotropic tail $|x|^{-(n+\gamma)}$, so $-\int p_{x_0}\log p_{x_0}\,dx < \infty$.

The decisive point is that no step in either bound invokes $\E[S]$, which diverges for $\gamma/2 < 1$; equivalently, neither invokes a finite stationary covariance. The covariance $\Sss$ may be infinite, while the snapshot mutual information is not.

\paragraph{Step (d). Independence of $\omega$.}
Under full isotropy $A = a I$, $D = D_0 I$, pass to the co-rotating frame $y_t = R(-\omega t)\,x_t$ with $R(-\omega t) \in SO(n)$ the planar rotation by angle $-\omega t$. This removes the antisymmetric part of the drift and maps Eq.~\eqref{eq:S-OU} to a rotation-free isotropic OU--L\'evy process for $y_t$, by Step~(a) with the same isotropic $\alpha$-stable law as the lab-frame process. The frame map $(x_0, x_t) \mapsto (x_0, R(-\omega t) x_t)$ is a measurable bijection (the inverse is $(y_0, y_t) \mapsto (y_0, R(\omega t) y_t)$), so mutual information is preserved: $I(x_0; x_t) = I(y_0; y_t)$, since $I(\phi(X); \psi(Y)) = I(X; Y)$ for any measurable bijections $\phi, \psi$. The latter is computed from a process with no $\omega$ in it, so the snapshot mutual information is exactly $\omega$-independent for all $\gamma \in (0, 2)$, by the identical orthogonality mechanism that gave Eq.~\eqref{eq:S-MI-iso} in the Gaussian case.

\paragraph{Technical assumptions.}
(1) The transition subordinator on $[0, t]$ is independent of the stationary pair $(x_0, S_0)$, so that $I((x_0, S_0); \mathcal S) = 0$. A representation in which $x_0$ and the transition share a single global scale would collapse $I(x_0; S_0)$ and $I^G$ into a single quantity but leave the log-moment-controlled finiteness bound of Eq.~\eqref{eq:S-finiteness-chain} structurally unchanged. (2) Existence and uniqueness of the stationary isotropic $\alpha$-stable OU law and of the joint snapshot law, in the Samorodnitsky-Taqqu and Sato senses, are assumed and standard in the literature. (3) The mutual information is the relative entropy $D(P_{x_0, x_t}\,\|\,P_{x_0} \otimes P_{x_t})$, and mutual absolute continuity is automatic under linear OU--L\'evy dynamics. (4) Finite logarithmic moments are a sufficient, not necessary, condition for $\E[I^G] < \infty$; the bound used here is intentionally conservative. (5) The $\omega$-independence is established for the isotropic instance of $[A,D] = 0$. The aligned anisotropic instance $[A,D] = 0$ is handled by the same spectral argument as in the Gaussian case (Sec.~\ref{app:co-reflect}), conditioned on the subordinator path.

\paragraph{Scope of the construction.}
The same construction does not extend to the entropy production rate. Its absolute scale in the $\gamma < 2$ regime requires a tempering regularization of the $\alpha$-stable jump kernel, because the heavy-tailed steady state is dominated by large jumps and the path-measure relative entropy carries the nonlocal jump kernel of the generator rather than the Gaussian quadratic form Eq.~\eqref{eq:S-EP-def}. The parity structure $\Omega = 0 \Leftrightarrow \Sigma = 0$, monotonicity in $|\omega|$, and the small-$\omega$ scaling $\Sigma \propto \omega^2$ is a symmetry statement and is independent of any such regularization. The parity structure is therefore regularization-free; the absolute scale is not.

\section{Misalignment perturbation around $[A,D] = 0$}\label{app:perturb}

The parity violation away from the structural alignment is controllable. Write the diffusion matrix as
\begin{equation}\label{eq:S-perturb-setup}
D \;=\; D_0 + \eta V, \qquad [A, D_0] = 0,
\end{equation}
with $\eta > 0$ a small misalignment parameter and $V$ a fixed symmetric matrix not necessarily commuting with $A$. To first order in $\eta$, the stationary covariance solving the Lyapunov equation~\eqref{eq:S-lyap} expands as
\[
\Sss(\omega) \;=\; \Sss^{(0)}(\omega) + \eta\,\Sss^{(1)}(\omega) + O(\eta^2),
\]
with $\Sss^{(0)}$ the aligned-case solution that satisfies the co-reflection identity~\eqref{eq:S-Sss-conjugate}. The first-order correction obeys
\begin{equation}\label{eq:S-Sss1}
(A - \Omega)\,\Sss^{(1)} \;+\; \Sss^{(1)}\,(A - \Omega)^\top \;=\; 2\,V.
\end{equation}

Let $R$ denote the reflection introduced in the even-branch analysis (Lemma~\ref{lem:reflection}), which preserves $A$ and $D_0$ and reverses $\Omega$ at the aligned point. Conjugating both sides of Eq.~\eqref{eq:S-Sss1} by $R$ on the left and $R^\top$ on the right gives
\[
(A + \Omega)\,\bigl(R\,\Sss^{(1)} R^\top\bigr) \;+\; \bigl(R\,\Sss^{(1)} R^\top\bigr)\,(A + \Omega)^\top \;=\; 2\,R\,V\,R^\top,
\]
which is the Lyapunov equation at $-\omega$ with right-hand side $2\,R V R^\top$.

When $V$ commutes with $R$ (the $R$-symmetric piece of the perturbation), the right-hand side coincides with $2 V$, so $R\,\Sss^{(1)}(\omega)\,R^\top$ solves the Lyapunov equation at $-\omega$, and uniqueness gives
\[
\Sss^{(1)}(-\omega) \;=\; R\,\Sss^{(1)}(\omega)\,R^\top.
\]
The first-order correction is then parity-even, and every spectral or log-determinant functional of the joint covariance remains exactly even at first order in $\eta$.

When $V$ fails to commute with $R$, decompose $V = V_+ + V_-$ with
\[
V_+ \;=\; \tfrac{1}{2}\bigl(V + R V R^\top\bigr), \qquad V_- \;=\; \tfrac{1}{2}\bigl(V - R V R^\top\bigr),
\]
the $R$-symmetric and $R$-antisymmetric parts respectively. The defining properties $R V_\pm R^\top = \pm V_\pm$ use the involution $R^2 = I$ from Lemma~\ref{lem:reflection}. The $V_+$ piece is parity-preserving by the argument above; the $V_-$ piece sources a parity-breaking first-order correction $\Sss^{(1,-)}$ that satisfies
\[
R\,\Sss^{(1,-)}(\omega)\,R^\top \;=\; -\Sss^{(1,-)}(-\omega).
\]
The magnitude of $V_-$ is controlled by the commutator $\|[V, R]\| \le 2\,\|V\|$, and at the aligned point the natural choice of $R$ makes $\|V_-\|$ proportional to the misalignment $\|[A, D]\|/(\|A\|\,\|D\|)$ to leading order.

Consequently the parity violation of any spectral or log-determinant functional of the joint covariance of $(x_0, x_t)$ scales linearly with $\|[A, D]\|$ for small misalignment. This is the sharp two-sided characterization referred to in the main text, and is visible numerically in the misaligned curve of FIG.~\ref{fig:even-odd}(a) of the main text, which is no longer symmetric about $\omega = 0$. Higher-order corrections refine the law but do not alter its leading scaling.

\end{document}